\documentclass{llncs}

\usepackage{amsmath}
\usepackage{amsfonts}
\usepackage{amssymb}
\usepackage{verbatim}
\usepackage{thanks}
\usepackage{latexsym}
\usepackage{times}

\newcommand{\oldbfe}[1]{\begin{bfseries}\emph{#1}\end{bfseries}}

\newcommand{\myra}{\mbox{$\:\rightarrow\:$}}

\newcommand{\fa}{\mbox{$\forall$}}

\newcommand{\LL}{\mbox{$\ldots$}}

\newcommand{\C}[1]{\mbox{$\{{#1}\}$}}           

\newcommand{\NI}{\noindent}

\newcommand{\II}{\vspace{2 mm}}

\newcommand{\szkew}[1]{\relax \setbox0=\hbox{\kern -24pt $\displaystyle#1$\kern 0pt }%
\box0}
{\catcode`\@=11 \global\let\ifjusthvtest@=\iffalse}

\newcounter{oldmycaption}

\newcommand{\eat}[1]{}

\newenvironment{Proof}
      {\medskip\noindent{\bf Proof.}}
      {\hfill$\Box$\medskip}
\newenvironment{proofof}[1]
      {\medskip\noindent{\bf Proof of #1.}}
      {\hfill$\Box$\medskip}
\def\smallromani{\renewcommand{\theenumi}{\roman{enumi}}
\renewcommand{\labelenumi}{(\theenumi)}}

\begin{document}

\author{Krzysztof Apt\inst{1,2} \and Vincent Conitzer\inst{3} \and Mingyu Guo\inst{3} \and Evangelos Markakis\inst{1}}

\institute{Centre for Math and Computer Science (CWI), Amsterdam, The Netherlands  \newline \email{\{apt, vangelis\}@cwi.nl} \and University of Amsterdam, Institute of Language, Logic and Computation, Amsterdam, The Netherlands \and Duke University, Department of Computer Science, Durham, NC, USA \newline 
\email{\{conitzer, mingyu\}@cs.duke.edu}
}

\title{Welfare Undominated Groves Mechanisms
}

\date{}
\maketitle
\begin{abstract}
A common objective in mechanism design is to choose the outcome
(for example, allocation of resources) that maximizes the sum of the
agents' valuations, without introducing incentives for agents to misreport
their preferences.  The class of Groves mechanisms achieves this; however,
these mechanisms require the agents to make payments, thereby reducing
the agents' total welfare. 

In this paper we introduce a measure for comparing two mechanisms with
respect to the final welfare they generate. This measure induces a partial
order on mechanisms and we study the question of finding minimal elements
with respect to this partial order. In particular, we say a non-deficit
Groves mechanism is {\em welfare undominated} if there exists no other
non-deficit Groves mechanism that always has a smaller or equal sum of
payments.  
We focus on two domains: (i) auctions with multiple identical
units and unit-demand bidders, and (ii) mechanisms for public project
problems.
In the first domain we analytically characterize all welfare undominated Groves
mechanisms that are anonymous and have linear payment functions, by showing
that the family of optimal-in-expectation linear redistribution mechanisms,
which were introduced in~\cite{GC08b} and include the Bailey-Cavallo
mechanism~\cite{Bailey97:Demand,Cavallo06:Optimal}, coincides with the family
of welfare undominated Groves mechanisms that are anonymous and linear in the setting we study. 
In the second domain we show that the classic VCG (Clarke) mechanism is welfare
undominated for the class of public project problems with equal participation
costs, but is not 
undominated for a more general class.  \end{abstract}




\section{Introduction}
\label{sec:intro}


Mechanism design is often employed for coordinating group decision making among
agents.  Often, such mechanisms impose taxes that agents have to pay to a
central authority. Although maximizing tax revenue is a desirable objective in
many settings (for example, if the mechanism is an auction designed by the
seller), it is not desirable in situations where no entity is profiting from
the taxes. Some examples include public project problems as well as certain resource allocation problems without a
seller (e.g., the right to use a shared good on a given time slot, or the exchange
of take-off slots among airline companies). In such
cases, we would like to have mechanisms that minimize the sum of the taxes (or,
even better, achieve budget balance, that is, the sum of the taxes is zero),
while maintaining other desirable properties, such as efficiency,
strategy-proofness and non-deficit (i.e., the mechanism does not need to be
funded by an external source).

The well-known VCG mechanism~\footnote{In this paper, ``the VCG mechanism''
  refers to the Clarke mechanism (aka pivotal mechanism), not to any other Groves mechanism.}  is
efficient, strategy-proof and incurs no deficit. More generally, the family
of Groves mechanisms, which includes VCG, is a family of efficient and
strategy-proof mechanisms. Unfortunately though, Groves mechanisms are not
budget balanced.  In fact, in sufficiently general settings, it is
impossible to have a mechanism that satisfies efficiency,
strategy-proofness, and budget balance~\cite{Green77:Characterization}.


We therefore consider the following problem: within the family of Groves
mechanisms, we want to identify non-deficit mechanisms that are optimal with
respect to the sum of the payments, i.e., we cannot lower the mechanism's
payments without violating efficiency, strategy-proofness or the non-deficit
property. Such a mechanism, in a sense, maximizes the agents' welfare
(among efficient mechanisms\footnote{By sacrificing efficiency, it is
  sometimes possible to drastically lower the payments, so that the net
  effect is an increase in the agents' welfare~\cite{Guo08:Better,Faltings05:Budget}. However, most of the prior work
  has focused on the case where efficiency is a hard constraint, and we
  will do so in this paper.}). To make this precise, we first introduce a
measure for comparing two feasible mechanisms
(mechanisms that are efficient, strategy-proof and satisfy the non-deficit property). 
We say that a feasible
Groves mechanism $M$ {\em welfare dominates} another feasible Groves
mechanism $M'$ if for every type vector of the agents, the sum of the
payments under $M$ is no more than the sum of the payments under $M'$, and
this holds with strict inequality for at least one type vector. This
definition induces a partial order on feasible Groves mechanisms and we
wish to identify minimal elements in this partial order. We call such
minimal elements {\em welfare undominated}.  Other partial orders, as well
as other notions of optimality, have recently been considered in other work
on redistribution mechanisms (see
Section~\ref{subsec:relwork}). The notion of optimality that we study here
is different from the previously studied ones at both a conceptual and
a technical level, as we illustrate below.


We study the question of finding welfare undominated mechanisms in two
domains.  The first is auctions of multiple identical units with
unit-demand bidders. In this setting, it is easy to see that VCG is welfare
dominated by other Groves mechanisms, such as the Bailey-Cavallo
mechanism~\cite{Bailey97:Demand,Cavallo06:Optimal}. We obtain a complete
characterization of linear and anonymous redistribution mechanisms that are
minimal elements in this partial order: we show that a linear, anonymous
Groves mechanism is welfare undominated if and only if it belongs to the
class of {\em Optimal-in-Expectation Linear (OEL) redistribution
  mechanisms}, which include the Bailey-Cavallo mechanism and were
introduced in~\cite{GC08b}. 
The second domain is public project problems, where a set of agents must
decide on financing a project (e.g., building a bridge). Here, we show
that in the case where the agents have identical participation costs, no
mechanism welfare dominates the VCG mechanism. On the other hand, when the
participation costs can be different across agents, there exist
mechanisms that welfare dominate VCG. In both domains, our proofs rely on
some general properties we establish for anonymous mechanisms, which may be
of independent interest (see Section~\ref{sec:anonymous}). 

The omitted proofs appear in the full version of the paper. 

\subsection{Related Work}
\label{subsec:relwork}

Recently, there has been a series of works on redistribution mechanisms,
which are Groves mechanisms that redistribute some of the VCG payment back
to the bidders. Bailey and Cavallo~\cite{Bailey97:Demand,Cavallo06:Optimal}
introduced a mechanism that welfare dominates VCG in some cases, such as
single-item auctions, but coincides with VCG in some more general settings.
We will refer to this mechanism as the BC mechanism from now on (in fact,
Bailey's mechanism is not always the same as Cavallo's mechanism, but it is
in the settings in which we study it).  A special case of the BC mechanism
was independently discovered by Porter {\em et al.}~\cite{Porter04:Fair}.
Cavallo also proved that the BC mechanism is optimal among the family of
{\em surplus-anonymous} mechanisms; however, this is a quite restrictive
class of mechanisms.
Guo and Conitzer~\cite{GC08j} solved for a worst-case optimal
redistribution mechanism for multi-unit auctions with nonincreasing
marginal values.  Moulin~\cite{Moulin07:Efficient} independently derived
the same mechanism under a slightly different worst-case optimality notion
(in the more restrictive setting of multi-unit auctions with unit demand
only).  These worst-case notions are different notions of optimality than
the one we consider in this paper. Guo and Conitzer~\cite{GC08b} also solve
for mechanisms that maximize expected redistribution (in a certain class of
mechanisms), when a prior is available.  Another notion of optimality,
which is closer to the one studied in this paper, was introduced
in~\cite{GC08a}, namely the notion of {\em undominated} mechanisms. A
mechanism is undominated if there is no other mechanism under which every
{\em individual} agent pays weakly less for every type vector, and strictly
less in at least one case.  This is a weaker concept than ours, in the sense
that for a mechanism that is undominated, there may still exist mechanisms
that welfare dominate it (by increasing the payment from some agents to
decrease the payments from other agents more). In the other direction, if a
mechanism is welfare undominated, then it is also undominated. We believe
that the notion we study in this paper is more appropriate when one is
interested in the final welfare of the agents.  Technically, welfare
undominance appears much more challenging and seems to require different
techniques.

\section{Preliminaries}\label{sec:prelim}

\subsection{Tax-based mechanisms}

We first briefly review tax-based mechanisms (see, e.g.,~\cite{MWG95}).
Assume that there is a set of possible outcomes or \oldbfe{decisions} $D$,
a set $\{1, \LL, n\}$ of players where $n \geq 2$, and for each player $i$
a set of \oldbfe{types} $\Theta_i$ and an (\oldbfe{initial})
\oldbfe{utility function} $ v_i : D \times \Theta_i \myra \mathbb{R}$. Let
$\Theta := \Theta_1 \times \cdots \times \Theta_n$.

In a (direct revelation) mechanism, each player reports a type $\theta_i$ and based on this, the mechanism selects an outcome and a payment to be made by every agent. Hence a mechanism is given by a pair of functions $(f,t)$, where $f$ is the decision function and $t = (t_1,...,t_n)$ is the tax function that determines the players' payments, i.e., $f: \Theta \myra D$, and $t: \Theta \myra \mathbb{R}^n$.


We assume that the (\oldbfe{final}) \oldbfe{utility function} for player $i$ is a function $u_i: D \times \mathbb{R}^n \times \Theta_i \myra \mathbb{R}$
defined by
$
u_i(d,t_1, \LL, t_n, \theta_i) :=  v_i(d, \theta_i) + t_i
$  (that is, utilities are {\em quasilinear}).
For each vector $\theta$ of announced types, if
$t_i(\theta) \geq 0$, player $i$ \oldbfe{receives} $t_i(\theta)$, and
if $t_i(\theta) < 0$, he \oldbfe{pays} $|t_i(\theta)|$.
Thus when the true type of player $i$ is $\theta_i$ and his
announced type is $\theta'_i$, his final utility is 
\[
u_i((f,t)(\theta'_i,
\theta_{-i}), \theta_i) = v_i(f(\theta'_i, \theta_{-i}), \theta_i) +
t_i(\theta'_i, \theta_{-i}),
\] 
where $\theta_{-i}$ are the types
announced by the other players.

\eat{


}
\subsection{Properties of tax-based mechanisms}

We say that a tax-based mechanism $(f,t)$ is
  \begin{enumerate}
  \item[$\bullet$] \oldbfe{efficient} if for all $\theta \in
\Theta$ and $d' \in D$, 
$\sum_{i = 1}^{n} v_i(f(\theta), \theta_i) \geq \sum_{i = 1}^{n} v_i(d', \theta_i)$, 
  \item[$\bullet$] \oldbfe{budget-balanced} if $\sum_{i = 1}^{n} t_i(\theta) = 0$ for all $\theta\in\Theta$,

  \item[$\bullet$] \oldbfe{feasible} if $\sum_{i = 1}^{n} t_i(\theta) \leq 0$ for all $\theta$, i.e., the mechanism does not need to be funded by an external source,

  \item[$\bullet$] \oldbfe{pay-only} if $t_i(\theta) \leq 0$ for all $\theta$ and all $i \in \{1, \LL, n\}$,
  
  \item[$\bullet$] \oldbfe{strategy-proof} 
if for all $\theta$, $i \in \C{1,\LL,n}$ and
$\theta'_i$,
\[
u_i((f,t)(\theta_i, \theta_{-i}), \theta_i) \geq
u_i((f,t)(\theta'_i, \theta_{-i}), \theta_i).
\]
  \end{enumerate}

  Tax-based mechanisms 
  can be compared in terms of the final social welfare they generate
  ($\sum_{i = 1}^{n} u_i((f,t)(\theta), \theta_i)$). More precisely, one
  can define the following two natural partial orders as a way to compare
  mechanisms. The first was introduced in \cite{GC08a}. The second
  is the concept that we introduce and study in this paper, which we believe
  is a more appropriate concept when one is interested in the final
  social welfare of the agents.

\begin{definition}
Given two tax-based mechanisms $(f,t)$ and $(f',t')$ we say
that $(f',t')$ \oldbfe{dominates} $(f,t)$ (due to \cite{GC08a})
if 
\begin{itemize}
\item[$\bullet$] for all $\theta \in \Theta$ and all $i \in \{1, \LL, n\}$,
$u_i((f,t)(\theta), \theta_i) \leq  u_i((f',t')(\theta), \theta_i),
$
\item[$\bullet$] for some $\theta \in \Theta$ and some $i \in \{1, \LL, n\}$,
$u_i((f,t)(\theta), \theta_i) <  u_i((f',t')(\theta), \theta_i).$
\end{itemize} 
\end{definition}

\begin{definition}
Given two tax-based mechanisms $(f,t)$ and $(f',t')$ we say
that $(f',t')$ \oldbfe{welfare dominates} $(f,t)$ 
if
\begin{itemize}
\item[$\bullet$] for all $\theta \in \Theta$,
$
\sum_{i = 1}^{n} u_i((f,t)(\theta), \theta_i) \leq \sum_{i = 1}^{n} u_i((f',t')(\theta), \theta_i),
$

\item[$\bullet$] for some $\theta \in \Theta$,
$
\sum_{i = 1}^{n} u_i((f,t)(\theta), \theta_i) < \sum_{i = 1}^{n} u_i((f',t')(\theta), \theta_i).
$
\end{itemize}
\end{definition}

In this paper, we are interested only in Groves mechanisms, so that the
decision function $f$ is always efficient, and (welfare) dominance
is strictly due to differences in the tax function $t$.  Specifically, in
this context we have that
$(f,t')$ dominates $(f,t)$ (or simply $t'$ dominates $t$)
if and only if

\begin{itemize}
\item[$\bullet$] for all $\theta \in \Theta$ and all $i \in \{1, \LL, n\}$,
$t_i(\theta) \leq  t'_i(\theta)$, and
\item[$\bullet$] for some $\theta \in \Theta$  and some $i \in \{1, \LL, n\}$,
$t_i(\theta) < t'_i(\theta)$,
\end{itemize}
and $t'$ welfare dominates $t$ if 

\begin{itemize}
\item[$\bullet$] for all $\theta \in \Theta$, $\sum_{i =
  1}^{n} t_i(\theta) \leq \sum_{i = 1}^{n} t'_i(\theta)$, and

\item[$\bullet$] for some $\theta \in \Theta$, $\sum_{i =
  1}^{n} t_i(\theta) < \sum_{i = 1}^{n} t'_i(\theta)$.

\end{itemize}

For two tax-based mechanisms $t, t'$, it is clear that if $t'$ dominates $t$, then it also welfare dominates $t$. 
The reverse implication, however, does not need to hold.\footnote{In
  Appendix~\ref{app:distinct}, we provide an example of two tax-based
  mechanisms that illustrates this.}
  
We now define a transformation on tax-based mechanisms originating from the
same decision function. This transformation was originally defined in
\cite{Bailey97:Demand} and \cite{Cavallo06:Optimal} for the specific case
of the VCG mechanism and in \cite{GC08a} for feasible Groves mechanisms.
We call it the \oldbfe{BCGC transformation} after the authors of these
papers.

Consider a tax-based mechanism $(f,t)$. Given $\theta =
(\theta_1,...,\theta_n)$, let $T(\theta)$ be the total amount of taxes,
i.e., $T(\theta) := \sum_{i=1}^{n} t_i(\theta)$.  For each $i \in
\{1,\LL,n\}$ let\footnote{To ensure that the maximum actually exists we
  assume that each tax function $t_i$ is continuous and each set of types
  $\theta_i$ is a compact subset of some $\mathbb{R}^k$.}
\[
S^{BCGC}_{i}(\theta_{-i}) := \max_{\theta_i' \in \Theta_i} T(\theta'_i, \theta_{-i}).
\]
We then define the tax-based mechanism
$t^{BCGC}$  as follows:
\[
t_i^{BCGC}(\theta) := t_i(\theta) - \frac{S^{BCGC}_i(\theta_{-i})}{n}.
\]

The following observations generalize some of the results of
\cite{Bailey97:Demand,Cavallo06:Optimal,GC08a}.

\begin{note} \label{not:bc}
\begin{enumerate} \smallromani
\item Each tax-based mechanism of the form $t^{BCGC}$ is feasible.  

\item If $t$ is feasible, then either $t$ and $t^{BCGC}$ coincide or  $t^{BCGC}$ dominates $t$.
\end{enumerate}
\end{note}

\subsection{Groves mechanisms}

Each \oldbfe{Groves mechanism} is a tax-based mechanism $(f,t)$ such that the following hold\footnote{Here and below $\sum_{j\not=i}$ is a shorthand for the
  summation over all $j \in \{1,\LL,n\}, \ j \not=i$.}:

\begin{itemize}

\item[$\bullet$] $f(\theta)\in \arg\max_d \sum_{i=1}^n v_i(d,\theta_i)$, 
i.e., the chosen outcome maximizes the initial social welfare.

\item[$\bullet$] $t_i : \Theta \myra \mathbb{R}$ is defined by $t_i(\theta) := g_i(\theta) + h_i(\theta_{-i})$,

\item[$\bullet$] $g_i(\theta) := \sum_{j \neq i} v_j(f(\theta), \theta_j)$,
  
\item[$\bullet$] $h_i: \Theta_{-i} \myra \mathbb{R}$ is an arbitrary function.

\end{itemize}




Intuitively, $g_i(\theta)$ represents the (initial) social welfare from the
decision $f(\theta)$, when player $i$'s (initial) utility is not counted.  
We now recall the following result (e.g.,~\cite{MWG95}):
\II

\NI \textbf{Groves Theorem} Every Groves mechanism $(f,t)$, is efficient and strategy-proof.  \II

For several decision problems the only efficient and strategy-proof tax-based
mechanisms are Groves mechanisms. By a general result of \cite{Hol79}
this is the case for both domains that we consider in this
paper and explains our focus on Groves mechanisms.

A feasible Groves mechanism is \oldbfe{undominated} if there is no
other feasible Groves mechanism that dominates it~\cite{GC08a}. A
feasible Groves mechanism is \oldbfe{welfare undominated} if there is
no other feasible Groves mechanism that welfare dominates it.  Welfare
undominance is a strictly stronger concept than undominance,
as is
illustrated in Appendix~\ref{app:distinct}.  

A special Groves mechanism---the \oldbfe{VCG} or \oldbfe{Clarke}
mechanism---is obtained using\footnote{Here and below, to ensure that the
  considered maximum exist, we assume that $f$ and each $v_i$ are continuous
  functions and $D$ and each $\theta_i$ are compact subsets of some
  $\mathbb{R}^k$.} 
  \[ h_i(\theta_{-i}) := - \max_{d \in D} \sum_{j \neq i}
v_j(d,\theta_j).\]
 In this case,
\[
t_i(\theta)  := \sum_{j \neq i} v_j(f(\theta), \theta_j) - \max_{d \in D} \sum_{j \neq i} v_j(d, \theta_j),
\]
which shows that the VCG mechanism is pay-only.

Following \cite{Cavallo06:Optimal}, let us now consider the mechanism that
results from applying the BCGC transformation to the VCG mechanism. We
refer to this as the Bailey-Cavallo mechanism or simply the BC mechanism. Let $\theta' :=
(\theta_1,...,\theta_{i-1}, \theta_i', \theta_{i+1},...,\theta_n)$, so
$\theta_j' = \theta_j$ for $j\neq i$ and the $i$th player's type in the
type vector $\theta'$ is $\theta'_i$.
Then
\[
S^{BCGC}_i(\theta_{-i}) = \max_{\theta_i' \in \Theta_i} \sum_{k=1}^n \left[    \sum_{j \neq k} v_j(f(\theta'), \theta'_j) - \max_{d \in D} \sum_{j \neq k} v_j(d, \theta'_j)\right],
\]
that is,
\begin{equation}
 \label{eq:BC}
S^{BCGC}_i(\theta_{-i}) = \max_{\theta_i'\in\Theta_i} \left[ (n-1) \sum_{k=1}^n  v_k(f(\theta'), \theta'_k) - \sum_{k=1}^n \max_{d \in D} \sum_{j \neq k} v_j(d, \theta'_j)\right].
\end{equation}
In many settings, we have that for all $\theta$ and for all $i$,
$S^{BCGC}_i(\theta_{-i}) = 0$,
and consequently the VCG and BC mechanisms coincide. Whenever they do not,
by Note \ref{not:bc}$(ii)$ BC dominates VCG.  This is the case for the
single-item auction, as it can be seen that there $S^{BCGC}_i(\theta_{-i}) =
-[{\theta_{-i}]}_2$, where $[\theta_{-i}]_2$ is the
second-highest bid among bids other than player $i$'s own bid.

\section{Anonymous Groves mechanisms}
\label{sec:anonymous}
Throughout this paper, we will be interested in a special class
of Groves mechanisms, namely,
anonymous Groves mechanisms.
We provide here some results about this class that we will utilize in later
sections.  We call a function $f: A^n \myra B$ \oldbfe{permutation
  independent} if for all permutations $\pi$ of $\{1, \LL, n\}$, $f = f
\circ \pi$.
Following \cite{Mou88} we call a Groves mechanism (determined by the vector of functions
$(h_1, \LL, h_n)$) \oldbfe{anonymous} if
\begin{itemize}
\item all type sets $\Theta_i$ are equal,

\item all functions $h_i$ coincide and each of them is permutation independent.
\end{itemize}
Hence, an anonymous Groves mechanism is uniquely determined by a single
function $h : \Theta^{n-1} \rightarrow \mathbb R$.

In general, the VCG mechanism is not anonymous. But it is anonymous when
all the type sets are equal and all the initial utility functions $v_i$
coincide. This is the case in both of the domains that we consider in this
paper.

For any $\theta \in \Theta$ and any permutation $\pi$ of $\{1, \LL, n\}$ we define $\theta^{\pi} \in \Theta$ by letting
\[
\theta^{\pi}_i := \theta_{\pi^{-1}(i)}.
\]

Denote by $\Pi(k)$ the set of all permutations of the set $\{1, \LL, k\}$.
Given a Groves mechanism $h:= (h_1, \LL, h_n)$ for which the type set $\Theta_i$ is the same
for every player (and equal to, say, $\Theta_0$)
we construct now a function $h': \Theta^{n-1}_0 \myra \mathbb R$ by putting
\[
h'(x) := \frac{\sum_{\pi \in \Pi(n-1)} \sum_{j = 1}^{n} h_j(x^{\pi})}{n!},
\]
where $x^{\pi}$ is defined analogously to $\theta^{\pi}$.

Note that $h'$ is permutation independent, so $h'$ is an anonymous Groves mechanism.

The following lemma shows that some of the properties of $h$ transfer to $h'$.

\begin{lemma} 
\label{lem:anon}
Consider a Groves mechanism $h$ and the corresponding anonymous Groves mechanism $h'$. Let $G(\theta) := \sum_{j = 1}^{n}  v_j(f(\theta), \theta_j)$. Suppose that for all permutations $\pi$ of $\{1, \LL, n\}$,
$G(\theta) = G(\theta^{\pi})$. Then:
  \begin{enumerate} \smallromani
  \item  If $h$ is feasible, so is $h'$.

  \item   If an anonymous Groves mechanism $h^0$ is welfare dominated by $h$, 
then it is welfare dominated by $h'$.
  \end{enumerate}
\end{lemma}

The assumption in
Lemma~\ref{lem:anon} of permutation independence of $G(\cdot)$ is satisfied in
both of the domains that we consider in this paper.
Basically, Lemma~\ref{lem:anon} says that if a Groves mechanism is not welfare
undominated, then it must be welfare dominated by an anonymous Groves
mechanism. 

\section{Multi-unit auctions with unit demand} \label{sec:multi-unit}
In this section, we consider auctions where there are multiple identical
units of a single good and all players have unit demand, i.e., each player
wants only one unit.  (When there is only one unit, we have a standard
single-item auction.) For this setting, we obtain an analytical
characterization of all welfare undominated Groves mechanisms that are
anonymous and have linear payment functions, by proving that the
optimal-in-expectation linear redistribution mechanisms (OEL
mechanisms)~\cite{GC08b}, which include the BC mechanism, are the only
welfare undominated Groves mechanisms that are anonymous and linear.  We
also show that undominance and welfare undominance are equivalent if we
restrict our consideration to Groves mechanisms that are anonymous and
linear in the setting of multi-unit auctions with unit demand.

\subsection{Optimal-in-expectation linear redistribution mechanisms}

The optimal-in-expectation linear redistribution mechanisms are special cases
of Groves mechanisms that are anonymous and linear. The OEL mechanisms are
defined only for {\em multi-unit auctions with unit demand}, in which
there are $m$ indistinguishable units for sale, and
no bidder is interested in obtaining more than one unit. For player $i$, her type $\theta_i$ is
her valuation for winning one unit. We assume all bids (announced types) are
bounded below by $L$ and above by $U$, i.e., $\Theta_i = [L,U]$.  ($L$ can be $0$.)

The tax function $t$ of an anonymous linear Groves mechanism is defined as
$t_i(\theta)=t_i^{VCG}(\theta)+r(\theta_{-i})$ for all $i$ and $\theta$.  Here
$t^{VCG}$ is (the tax function of) the VCG mechanism, and $r$ is a linear
function defined as
$r(\theta_{-i})=c_0+\sum\limits_{j=1}^{n-1}c_j[\theta_{-i}]_j$ (where
$[\theta_{-i}]_j$ is the $j$th highest bid among $\theta_{-i}$).  For OEL, the
$c_j$'s are chosen according to one of the following options (indexed by $k$,
$k$ is from $0$ to $n$, and $k-m$ is odd):\\

%
%

$\mathbf{k=0}$: 

$c_i=(-1)^{m-i}{n-i-1 \choose n-m-1}/{m-1 \choose i-1}$ for $i=1,\ldots,m$, 

$c_0=Um/n-U\sum_{i=1}^m(-1)^{m-i}{n-i-1 \choose n-m-1}/{m-1\choose i-1}$, and $c_i=0$ for other $i$.

$\mathbf{k=1,2,\ldots,m}$:

        $c_i=(-1)^{m-i}{n-i-1 \choose n-m-1}/{m-1 \choose i-1}$ for $i=k+1,\ldots,m$,

  $c_k=m/n-\sum_{i=k+1}^m(-1)^{m-i}{n-i-1 \choose n-m-1}/{m-1\choose i-1}$, and $c_i=0$ for other $i$.

$\mathbf{k=m+1,m+2,\ldots,n-1}$: 

$c_i=(-1)^{m-i-1}{i-1\choose
            m-1}/{n-m-1\choose n-i-1}$ for $i=m+1,\ldots,k-1$,

          $c_k=m/n-\sum_{i=m+1}^{k-1}(-1)^{m-i-1}{i-1\choose
            m-1}/{n-m-1\choose n-i-1}$, and $c_i=0$ for other $i$.

$\mathbf{k=n}$: 

$c_i=(-1)^{m-i-1}{i-1\choose
            m-1}/{n-m-1\choose n-i-1}$ for $i=m+1,\ldots,n-1$,

          $c_0=Lm/n-
                  L\sum_{i=m+1}^{n-1}(-1)^{m-i-1}{i-1\choose
            m-1}/{n-m-1\choose n-i-1}$, and $c_i=0$ for other $i$.\\

%

For example, when $k=m+1$, we have $c_{m+1}=m/n$ and $c_i=0$ for all other $i$.  For
this specific OEL mechanism,
$t_i^{OEL}(\theta)=t_i^{VCG}(\theta)+\frac{m}{n}[\theta_{-i}]_{m+1}$.  That is,
besides paying the VCG payment, every player receives an amount that is equal to
$m/n$ times the $(m+1)$th highest bid from the other players.  Actually, this
is the BC mechanism for this setting.

One property of the OEL mechanisms is that the sum of the taxes
$\sum_{i=1}^nt^{OEL}_i(\theta)$ is always less than or equal to $0$ and it
equals $0$ whenever 

$\bullet$ $[\theta]_1=U$, if $k=0$.

        $\bullet$ $[\theta]_{k+1}=[\theta]_k$, if $k \in \{1, \LL,  n-1\}$.

        $\bullet$ $[\theta]_n=L$, if $k=n$.

Using this property, we will prove that the OEL mechanisms are the only welfare
undominated Groves mechanisms that are anonymous and linear.

\subsection{Characterization of welfare undominated Groves mechanisms that are anonymous and linear}


We first show that the OEL mechanisms are welfare undominated.  (It has
previously been shown that they are undominated~\cite{GC08a}, but
as we pointed out, being welfare undominated is  a stronger property.)

\begin{theorem} 
\label{thm:oel}
No feasible Groves mechanism welfare dominates an OEL
mechanism.
\end{theorem}

According to Lemma~\ref{lem:anon}, we only need to prove this for the case
of anonymous Groves mechanisms:

\begin{lemma} 
\label{lem:oel}
No feasible anonymous Groves mechanism  welfare dominates an OEL mechanism.
\end{lemma}

We now show that within the family of anonymous and linear Groves
mechanisms, the OEL mechanisms are the only ones that are welfare
undominated.  Actually, they are also the only ones that are undominated,
which is a stronger claim since being undominated is a weaker property.

\begin{theorem} \label{thm:characterize} If a feasible anonymous linear Groves
mechanism is undominated, then it must be an OEL mechanism.
\end{theorem}

Hence, we
have the following complete characterization in this context:

\begin{corollary}
  A feasible anonymous linear Groves mechanism is (welfare) undominated if
  and only if it is an OEL mechanism.
\end{corollary}

%

The above corollary also shows that if we consider only Groves mechanisms that
are anonymous and linear in the setting of multi-unit auctions with unit demand,
then undominance and welfare undominance are 
equivalent.\footnote{Thus, we have also characterized all
  undominated Groves mechanisms that are anonymous and linear. There is no
  corresponding result in \cite{GC08a}.}

\section{Public project problem with equal participation costs}

We now study a well known class of decision problems, namely public project problems---see, e.g.,
\cite{MWG95,Mou88,Moo06}.

\paragraph{Public project problem.} \label{exa:public1}
\mbox{} 
\NI
Consider
$
(D, \Theta_1, \LL, \Theta_n, v_1, \LL, v_n),
$
where 

\begin{itemize}
\item[$\bullet$] $D = \{0, 1\}$
(reflecting whether a project is canceled or takes place),

\item[$\bullet$] for all $i \in \{1, \LL, n\}$,
$\Theta_i = [0,c]$, where $c > 0$,

\item[$\bullet$] for all $i \in \{1, \LL, n\}$, $v_i(d, \theta_i) := d (\theta_i - \frac{c}{n})$,

\eat{
\item[$\bullet$] $
        f(\theta) :=
        \left\{
        \begin{array}{l@{\extracolsep{3mm}}l}
        1    & \mathrm{if}\  \sum_{i = 1}^{n} \theta_i \geq c \\
        0       & \mathrm{otherwise}
        \end{array}
        \right.
$
}
\end{itemize}

In this setting a set of $n$ agents needs to decide on financing a project
of cost $c$.  In the case that the project takes place, each agent
contributes the same share, $c/n$, so as to cover the total cost.  Hence
the participation costs of all players are the same. When the players
employ a tax-based mechanism to decide on the project, then in addition to
$c/n$, each player also has to pay or receive the tax, $t_i(\theta)$,
imposed by the mechanism.

By the result of Holmstrom~\cite{Hol79}, the only efficient and
strategy-proof tax-based mechanisms in this domain are Groves mechanisms.
To determine the efficient outcome for a given type vector $\theta$, note
that $\sum_{i = 1}^{n} v_i(d, \theta_i) = d(\sum_{i = 1}^{n} \theta_i -
c)$. Hence efficiency here for a mechanism $(f,t)$ means that $f(\theta) =
1$ if $\sum_{i = 1}^{n} \theta_i \geq c$ and $f(\theta) = 0$ otherwise,
i.e., the project takes place if and only if the declared total value that
the agents have for the project exceeds its cost.  We first observe the
following result.

\begin{note}
In the public project problem the BC mechanism coincides with VCG.
\end{note}

\begin{Proof}
It suffices to check that in equation (\ref{eq:BC}) it holds that  $S_i^{BCGC}$ $(\theta_{-i}) = 0$ for all $i$ and all $\theta_{-i}$. 
By the feasibility of VCG we have 
 $S_i^{BCGC} \leq 0$, hence all we need is to show that there is a value for $\theta_i'$ that makes 
the expression in  (\ref{eq:BC}) equal to $0$. Checking this is quite simple.
If $\sum_{j \neq i} \theta_j < \frac{n-1}{n} c$, 
then we take $\theta_i' := 0$ and otherwise
$\theta_i' := c$.
\end{Proof}

We now show that in fact VCG cannot be improved upon.
Before stating our result, we would like to note that one ideally would like to have a mechanism that is budget-balanced, i.e., $\sum_i t_i(\theta) = 0$ for all $\theta$, so that in total the agents only pay the cost of the project and no more. However this is not possible and as explained in~\cite[ page 861-862]{MWG95}, for the public
project problem no mechanism exists that is efficient,
strategy-proof and budget balanced. 
Our theorem below considerably strengthens this result, showing that VCG is optimal with respect to minimizing the total payment of the players.

\begin{theorem}  \label{thm:feasible}
  In the public project problem there exists no feasible Groves mechanism
  that welfare dominates the VCG mechanism.
\end{theorem}

As in Section~\ref{sec:multi-unit}, we first establish the desired conclusion for anonymous Groves mechanisms and then extend it to arbitrary ones by Lemma~\ref{lem:anon}.

\begin{lemma}
\label{anonymous}
In the public project problem there exists no anonymous feasible Groves
mechanism that welfare dominates the VCG mechanism.
\end{lemma}




\section{Public project problem: the general case}
\label{sec:public1}

The assumption that we have made so far in the public project problem that
each player's cost share is the same may not always be realistic. Indeed, it
may be argued that `richer' players (read: larger enterprises) should
contribute more.  Does it matter if we modify the formulation of the
problem appropriately? The answer is `yes'.  First, let us formalize this
problem.  We assume now that each (initial) utility function is of the form
$v_i(d, \theta_i) := d (\theta_i - c_i)$,
where for all $i \in \C{1, \LL, n}$, \ $c_i > 0$ and $\sum_{i = 1}^{n} c_i = c$.

In this setting, $c_i$ is the cost share of the project cost to be
financed by player $i$. 
We call the resulting problem the \oldbfe{general public project problem}.
It is taken from \cite[page 518]{Moo06}.  We first prove the following
optimality result concerning the VCG mechanism.

\begin{theorem} \label{thm:opt}
In the general public project problem there is no pay-only Groves mechanism that dominates the VCG mechanism.
\end{theorem}

It remains an open problem whether the above result can be extended to the
welfare dominance relation.  On the other hand, the above theorem cannot be
extended to feasible Groves mechanisms, as the following result holds.

\begin{theorem}
\label{thm:instanceexists}
For any $n\geq 3$, an instance of the general public project problem with $n$ players exists for which
the BC mechanism dominates the VCG mechanism.
\end{theorem}

By Theorem \ref{thm:opt}, the BC mechanism in the proof of
the above theorem is not pay-only. 


\section{Summary}

In this paper, we introduced and studied the following relation on feasible
Groves mechanisms: a feasible Groves mechanism {\em welfare dominates}
another feasible Groves mechanism if the total welfare (with taxes taken
into account) under the former is at least as great as the total welfare
under the latter, for any type vector---and the inequality is strict for at
least one type vector.  This dominance notion is different from the one
proposed in~\cite{GC08a}.
We then studied welfare (un)dominance in two domains.  The first domain we
considered was that of auctions with multiple identical units and unit
demand bidders.  In this domain, we analytically characterized all welfare
undominated Groves mechanisms that are anonymous and have linear payment
functions.
The second domain we considered is that of public
project problems.  In this domain, we showed that the VCG mechanism is welfare
undominated if cost shares are equal, but also that this is not necessarily
true if cost shares are not necessarily equal (though we showed that the VCG
mechanism remains undominated in the weaker sense of~\cite{GC08a} among 
pay-only mechanisms in this more general setting).

\section*{Acknowledgments}

Guo and Conitzer thank the National Science Foundation and the Alfred
P.~Sloan Foundation for support (through award number IIS-0812113 and a
Research Fellowship, respectively). The work of Vangelis Markakis was
funded by the NWO project DIACoDeM, No 642.066.604.

\bibliography{references}

\begin{thebibliography}{10}

\bibitem{Bailey97:Demand}
M.~J. Bailey.
\newblock The demand revealing process: to distribute the surplus.
\newblock {\em Public Choice}, 91:107--126, 1997.

\bibitem{Cavallo06:Optimal}
R.~Cavallo.
\newblock Optimal decision-making with minimal waste: Strategyproof
  redistribution of {VCG} payments.
\newblock In {\em International Conference on Autonomous Agents and Multi-Agent
  Systems (AAMAS)}, pages 882--889, Hakodate, Japan, 2006.

\bibitem{Faltings05:Budget}
B.~Faltings.
\newblock A budget-balanced, incentive-compatible scheme for social choice.
\newblock In {\em Agent-Mediated Electronic Commerce (AMEC), LNAI, 3435}, pages
  30--43, 2005.

\bibitem{Green77:Characterization}
J.~Green and J.-J. Laffont.
\newblock Characterization of satisfactory mechanisms for the revelation of
  preferences for public goods.
\newblock {\em Econometrica}, 45:427--438, 1977.

\bibitem{Guo08:Better}
M.~Guo and V.~Conitzer.
\newblock Better redistribution with inefficient allocation in multi-unit
  auctions with unit demand.
\newblock In {\em Proceedings of the ACM Conference on Electronic Commerce
  (EC)}, Chicago, IL, USA, 2008.

\bibitem{GC08b}
M.~Guo and V.~Conitzer.
\newblock Optimal-in-expectation redistribution mechanisms.
\newblock In {\em AAMAS 2008: Proc. of 7th Int. Conf. on Autonomous Agents and
  Multi Agent Systems}, 2008.

\bibitem{GC08a}
M.~Guo and V.~Conitzer.
\newblock Undominated {VCG} redistribution mechanisms.
\newblock In {\em AAMAS 2008: Proc. of 7th Int. Conf. on Autonomous Agents and
  Multi Agent Systems}, 2008.

\bibitem{GC08j}
M.~Guo and V.~Conitzer.
\newblock Worst-case optimal redistribution of {VCG} payments in multi-unit
  auctions.
\newblock {\em Games and Economic Behavior}, To appear. Earlier version in EC
  07.

\bibitem{Hol79}
B.~Holmstrom.
\newblock Groves' scheme on restricted domains.
\newblock {\em Econometrica}, 47(5):1137--1144, 1979.

\bibitem{MWG95}
A.~Mas-Collel, M.~Whinston, and J.~Green.
\newblock {\em Microeconomic Theory}.
\newblock Oxford University Press, 1995.

\bibitem{Moo06}
J.~Moore.
\newblock {\em General Equilibrium and Welfare Economics: An Introduction}.
\newblock Springer, 2006.

\bibitem{Mou88}
H.~Moulin.
\newblock {\em Axioms of Cooperative Decision Making}.
\newblock Cambridge University Press, 1988.

\bibitem{Moulin07:Efficient}
H.~Moulin.
\newblock Efficient, strategy-proof and almost budget-balanced assignment,
  March 2007.
\newblock Working Paper.

\bibitem{Porter04:Fair}
R.~Porter, Y.~Shoham, and M.~Tennenholtz.
\newblock Fair imposition.
\newblock {\em Journal of Economic Theory}, 118:209--228, 2004.
\newblock Early version appeared in IJCAI-01.

\end{thebibliography}
\bibliographystyle{abbrv}


%
%
%


\appendix
\section{Dominance is distinct from welfare dominance}
\label{app:distinct}

In this appendix, we give two tax-based mechanisms $t$ and $t'$ (both
feasible, anonymous Groves mechanisms) such that
$t'$ welfare dominates $t$, but $t'$ does not dominate $t$. 
Consider a single-item auction with $4$ players.  We assume that for each
player, the set of allowed bids is the same, namely, integers from
$0$ to $3$.
Let $t^{VCG}$ be (the tax function of) the VCG mechanism.  For all
$\theta\in \{0,1,2,3\}^4$, $\sum_{i=1}^4t^{VCG}_i(\theta)=-[\theta]_2$.
This is because for a single-item auction, the VCG mechanism is the
second-price auction. 
We define $t$ and $t'$ as follows:
{\bf Function $t$:} 
For all $\theta$, $t_i(\theta) := t^{VCG}_i(\theta)+h(\theta_{-i})$, where
$h(\theta_{-i})=r([\theta_{-i}]_1,[\theta_{-i}]_2,[\theta_{-i}]_3)$, and
the function $r$ is given in the table below. (We recall that
$[\theta_{-i}]_j$ is the $j$th-highest bid among bids other than $i$'s own
bid.)
{\bf Function $t'$:}
For all $\theta$, $t'_i(\theta) := t^{VCG}_i(\theta)+h'(\theta_{-i})$,
where
$h'(\theta_{-i})=r'([\theta_{-i}]_1,[\theta_{-i}]_2,[\theta_{-i}]_3)$, and
the function $r'$ is given in the table below.


\begin{footnotesize}
\noindent
        \begin{tabular}{|c|c|c|c|}
\hline
$\mathbf{r(0,0,0)}$ & $0$ & $\mathbf{r'(0,0,0)}$ & $0$\\
\hline
$\mathbf{r(1,0,0)}$ & $0$ & $\mathbf{r'(1,0,0)}$ & $0$\\
\hline
$\mathbf{r(1,1,0)}$ & $1/4$ & $\mathbf{r'(1,1,0)}$ & $1/4$\\
\hline
$\mathbf{r(1,1,1)}$ & $1/4$ & $\mathbf{r'(1,1,1)}$ & $1/4$\\
\hline
$\mathbf{r(2,0,0)}$ & $0$ & $\mathbf{r'(2,0,0)}$ & $0$\\
\hline
$\mathbf{r(2,1,0)}$ & $1/12$ & $\mathbf{r'(2,1,0)}$ & $7/24$\\
\hline
$\mathbf{r(2,1,1)}$ & $0$ & $\mathbf{r'(2,1,1)}$ & $1/6$\\
\hline
\end{tabular}
        \begin{tabular}{|c|c|c|c|}\hline
                        $\mathbf{r(2,2,0)}$ & $1/2$ & $\mathbf{r'(2,2,0)}$ & $1/2$\\
\hline
$\mathbf{r(2,2,1)}$ & $0$ & $\mathbf{r'(2,2,1)}$ & $1/4$\\
\hline
$\mathbf{r(2,2,2)}$ & $1/2$ & $\mathbf{r'(2,2,2)}$ & $1/2$\\
\hline
$\mathbf{r(3,0,0)}$ & $0$ & $\mathbf{r'(3,0,0)}$ & $0$\\
\hline
$\mathbf{r(3,1,0)}$ & $1/4$ & $\mathbf{r'(3,1,0)}$ & $1/4$\\
\hline
$\mathbf{r(3,1,1)}$ & $0$ & $\mathbf{r'(3,1,1)}$ & $1/4$\\
\hline
$\mathbf{r(3,2,0)}$ & $2/3$ & $\mathbf{r'(3,2,0)}$ & $2/3$\\
\hline
\end{tabular}
        \begin{tabular}{|c|c|c|c|}\hline
                        $\mathbf{r(3,2,1)}$ & $1$ & $\mathbf{r'(3,2,1)}$ & $19/24$\\
\hline
$\mathbf{r(3,2,2)}$ & $0$ & $\mathbf{r'(3,2,2)}$ & $1/6$\\
\hline
$\mathbf{r(3,3,0)}$ & $2/3$ & $\mathbf{r'(3,3,0)}$ & $5/6$\\
\hline
$\mathbf{r(3,3,1)}$ & $0$ & $\mathbf{r'(3,3,1)}$ & $7/12$\\
\hline
$\mathbf{r(3,3,2)}$ & $1$ & $\mathbf{r'(3,3,2)}$ & $5/6$\\
\hline
$\mathbf{r(3,3,3)}$ & $0$ & $\mathbf{r'(3,3,3)}$ & $1/2$\\
\hline
& & & \\
\hline
\end{tabular}
\end{footnotesize}

With the above characterization, $t'$ welfare dominates $t$ (the total tax
under $t'$ is never lower, and in some cases it is strictly higher: for
example, for the bid vector $(3,2,2,2)$, the sum of the $r_i$ is $1/2$, but the
sum of the $r'_i$ is $1$).  On the other hand, $t'$ does not dominate $t$: for
example, $r(3,3,2) = 1 > 5/6 = r'(3,3,2)$.  In fact,  no feasible Groves
mechanism dominates $t$.

\end{document}